\documentclass[centertags,12pt]{amsart}
\usepackage{latexsym}
\usepackage{amssymb}

\newtheorem*{theorem*}{Theorem}
\newtheorem{theorem}{Theorem}[section]
\newtheorem{proposition}[theorem]{Proposition}
\newtheorem{corollary}[theorem]{Corollary}
\newtheorem{lemma}[theorem]{Lemma}

\theoremstyle{remark}
\newtheorem{remark}[theorem]{Remark}

\begin{document}

\title[Distances to codewords of a GRM code of order $1$]
{Hamming distances from a function to all codewords of a Generalized 
Reed-Muller code of order one}

\author{Miriam Abd\'on}
\address{IME, Univ. Federal Fluminense, Rua M´ario Santos Braga s/n, CEP 24.020-140, 
Niter´oi, Brazil}
\email{miriam@mat.uff.br}

\author{Robert Rolland}
\address{Universit\'e d'Aix-Marseille,Institut de Math\'{e}matiques
de Marseille, case 907, F13288 Marseille cedex 9, France}
\email{robert.rolland@acrypta.fr}
\date{\today}

\begin{abstract} 
For any finite field ${\mathbb F}_q$ with $q$ elements, we study the 
set ${\mathcal F}_{(q,m)}$ of functions from ${\mathbb F}_q^m$ into ${\mathbb F}^q$.
We introduce a transformation that allows us to determine
a linear system of $q^{m+1}$ equations and $q^{m+1}$
unknowns, which has for solution the Hamming distances of a function in ${\mathcal F}_{(q,m)}$
to all the affine functions. 
\end{abstract}

\maketitle

\section{Introduction}\label{s1}
\subsection{Generalized Reed-Muller codes of order $1$}
Let ${\mathbb F}_{q}$ be the finite field with $q$ elements. For any integer
$m\ge 1$, we will identify ${\mathbb F}_{q^m}$ with ${\mathbb F}_q^m$ as follows:
consider a basis $\{e_1,\cdots,e_m\}$ of ${\mathbb F}_{q^m}$ over ${\mathbb F}_q$,
then an element $u \in {\mathbb F}_{q^m}$ will be identified with the
vector $(u_1,\cdots,u_m)\in {\mathbb F}_q^m$ if and only if,
$u=\sum_{i=1}^{m}u_ie_i$.

If $u=(u_1,\cdots,u_m)$ and $v=(v_1,\cdots,v_m)$, are two elements
of ${\mathbb F}_q^m$ we will denote by $u.v$ their product in the field
${\mathbb F}_{q^m}$ and by $\langle u,v \rangle$ their scalar product
$$\langle u,v \rangle = \sum_{i=1}^m u_iv_i.$$

We denote by ${\mathcal F}_{(q,m)}=\{f:{\mathbb F}_q^m \to {\mathbb F}_{q} \}$ the set of
functions from ${\mathbb F}_q^m$ to ${\mathbb F}_{q}$. Each function 
$f \in {\mathcal F}_{(q,m)}$ can be identified with its image
$\left(\strut f(u)\right)_{u\in {\mathbb F}_q^m}$. We know that these  functions are polynomial functions
of $m$ variables. The kernel of the map  which associates to any polynomial the corresponding 
polynomial function is the ideal $I$ generated by the $m$ polynomials $X_i^q - X_i$.
The reduced polynomials are the polynomials $P(X_1,\cdots, X_m)$ such that for each $i$, 
the partial degree $\deg_i(P(X_1,\cdots, X_m))$
of  $P(X_1,\cdots, X_m)$ with respect to the variable $X_i$ is $\leq q-1$.
Then for any $f \in {\mathcal F}_{(q,m)}$ there exists an unique reduced polynomial 
$P(X_1,\cdots, X_m)$
for which $f$ is the associated polynomial function. The total degree of
$P(X_1,\cdots, X_m)$ is called the degree of $f$ and denoted by $\deg(f)$.

With these notations, the Generalized Reed-Muller code of order $1$
is the set
$$
RM^{(1)}_{(q,m)}=\{(g(u))_{u \in {\mathbb F}_q^m}~\mid~ g \in
{{\mathcal F}}_{(q,m)} \hbox{ and } \deg(g) \leq 1  \}.
$$

If $f, g \in {\mathcal F}_{(q,m)}$, the Hamming distance between these two
functions is defined by
$$
d(f,g)={\rm card}\left(\{u \in {\mathbb F}_{q}^m \mid f(u)\neq g(u)\}\right).
$$

\subsection{Organization of the article}

In this article we study the Hamming distances from a function $f\in {\mathcal F}_{(q,m)}$
to all the codewords $g \in RM^{(1)}_{(q,m)}$.

\section{An adapted transform}\label{sec:transform}

It is known that
every codeword $g \in RM^{(1)}_{(q,m)}$ can be characterized by a pair
$(v, t) \in {\mathbb F}_q^m \times {\mathbb F}_{q}$ in the following sense:

$$
g(u)=\langle u,v \rangle +t \quad \quad \forall u \in {\mathbb F}_q^m .
$$

If $f \in {\mathcal F}_{(q,m)}$ and $g$ as above, we have that
$$
d(f,g)={\rm card}\left(\{u \in {\mathbb F}_{q}^m \mid f(u)\neq \langle u,v \rangle
+t\}\right)=q^m-N_{v,t}(f),
$$
where $N_{v,t}(f)={\rm card}\left(\{u \in {\mathbb F}_{q}^m \mid f(u)= \langle u,v \rangle
+t\}\right)$.

Now the problem is to study the integer numbers $N_{(v,t)}(f)$. In order
to do that, we will introduce a transform on the group algebra ${\mathbb C}{\mathbb F}_{q}$
of the additive group ${\mathbb F}_q$
over the complex field ${\mathbb C}$
which is quite similar to a Fourier Transform.

More precisely, ${\mathbb C} {\mathbb F}_q$ is the algebra of formal linear combinations with coefficients
in ${\mathbb C}$
$$\sum_{t \in {\mathbb F}_q} \alpha _t Z^t$$
where the operations are defined by
$$\sum_{t \in {\mathbb F}_q} \alpha _t Z^t+\sum_{t \in {\mathbb F}_q} \beta _t Z^t=
\sum_{t \in {\mathbb F}_q} (\alpha _t+\beta _t) Z^t,$$
$$\lambda(\sum_{t \in {\mathbb F}_q} \alpha _t Z^t)=
\sum_{t \in {\mathbb F}_q} (\lambda \alpha _t) Z^t,$$
$$(\sum_{t \in {\mathbb F}_q} \alpha _t Z^t)(\sum_{t \in {\mathbb F}_q} \beta _t Z^t)=
\sum_{t \in {\mathbb F}_q} \biggl (\sum_{r+s=t}(\alpha _r \beta _s)\biggr )
Z^t.$$

\medskip

Let  ${{\mathcal G}}_{(q,m)}$ be the algebra of functions from ${\mathbb F}_q^m$ (or
from ${\mathbb F}_{q^m}$) into ${\mathbb C} {\mathbb F}_q$. 
It is a vector space of dimension $q^{m+1}$ over ${\mathbb C}$.
Let us define an order on ${\mathbb F}_{q^m} \times {\mathbb F}_q$ 
and define the family $(e_{u,t})_{(u,t)\in {\mathbb F}_{q^m} \times {\mathbb F}_q}$
of elements of ${{\mathcal G}}_{(q,m)}$ where
\begin{equation}\label{base}
e_{u,t}(v)=
\left \{
\begin{array}{lll}
0 & \hbox{ if } & v \neq u \\
Z^t & \hbox{ if } & v= u 
\end{array}
\right.
\end{equation}
This family is a basis of ${{\mathcal G}}_{(q,m)}$ and has $q^{m+1}$ elements. 

\medskip

Define the operator $T_{(q,m)}$ of the ${\mathbb C}$-vector space
${{\mathcal G}}_{(q,m)}$ by
$$T_{(q,m)}(\phi)(v) =\sum_{u \in {\mathbb F}_q^m}\phi(u) Z^{-\langle u,v \rangle}.$$

\begin{remark}
In the case where the function $\phi$ is given by
$\phi(v)=Z^{f(v)}$ for some $f \in {\mathcal F}_{(q,m)}$, then the
transform introduced above is the same that the the one introduced
by Ashikhmin and Litsyn (see \cite{a-l}). We recall here some
basic properties of this transform, for more details see
\cite{rol}.
\end{remark}

\begin{lemma}\label{evt}
The transform of $e_{u,t}$ by $T_{(q,m)}$ is given by
$$\epsilon_{u,t}(v)=T_{(q,m)}(e_{u,t})(v)=\sum_{w\in{\mathbb F}_q^m}e_{u,t}(w) Z^{-\langle
w,v \rangle}=Z^{t-\langle
u,v \rangle},$$
then
$$\epsilon_{u,t}=\sum_{(v,\tau)\in E_{-u,t}}e_{v,\tau},$$
where $E_{-u,t}$ is the hyperplane of ${\mathcal F}_{(q,m)}\times {\mathcal F}_q$ defined by
$$E_{-u,t}=\{ (v,\tau)\in {\mathbb F}_{q^m}\times {\mathbb F}_q ~|~ \tau=t-<u,v>\}.$$
\end{lemma}

\begin{lemma}\label{gamma}
Let $\gamma_a \in {\mathcal G}_{(q,m)}$ be defined by
$\gamma_a(u)=Z^{\langle a,u \rangle}$, then the transform of
$\gamma_a$ is given by
$$
T_{(q,m)}(\gamma_a)(v)= \begin{cases}
q^mZ^0 \quad&\text{if }\, v=a \\
q^{m-1}\sum_{t \in {\mathbb F}_q}Z^t \quad&\text{if }\, v\neq a.
\end{cases}
$$
\end{lemma}

\begin{proof}
We have successively
\begin{align*}
T_{(q,m)}(\gamma_a)(v)&=\sum_{u \in {\mathbb F}_q^m}\gamma_a(u)Z^{-\langle
u,v \rangle}\\
&=\sum_{u \in {\mathbb F}_q^m}Z^{\langle a,u \rangle}Z^{-\langle u,v
\rangle}\\
&=\sum_{u \in {\mathbb F}_q^m}Z^{\langle a-v,u \rangle}.
\end{align*}

If $v=a$ we have that $T_{(q,m)}(\gamma_a)(v)=q^mZ^0$ and then, when $v\neq a$, for
each $t \in {\mathbb F}_q$, the equation $\langle a-v,u \rangle=t$ defines
a hyperplane and consequently has $q^{m-1}$ solutions.
\end{proof}

\medskip

Let  $\phi$ be an element of ${{\mathcal G}}_{(q,m)}$, we denote by $\psi =
T_{(q,m)}(\phi)$ its transform, and by $\theta=T_{(q,m)}(\psi)$
its double transform.

\begin{theorem}\label{double}
With the previous notations we have
$$\theta (w)= q^{m-1}\sum_{t \in {\mathbb F}_q}(Z^0-Z^t)  \phi(-w)
+\left( q^{m-1}\sum_{t \in {\mathbb F}_q}Z^t\right) \psi (0).$$
\end{theorem}

\begin{proof}
We have that
\begin{align*}
\theta (w)&=\sum_{v \in {\mathbb F}_q^m}\left ( \sum_{u \in {\mathbb F}_q^m} \phi(u)
Z^{-\langle u,v \rangle}\right) Z^{-\langle v,w \rangle}\\
&=\sum_{v \in {\mathbb F}_q^m} \sum_{u \in {\mathbb F}_q^m} \phi(u) Z^{-\langle
u+w,v \rangle}\\
&=\sum_{u \in {\mathbb F}_q^m}\phi(u)  \sum_{v \in {\mathbb F}_q^m}
 Z^{-\langle u+w,v \rangle}
 \end{align*}

From Lemma \ref{gamma} we obtain
$$\theta (w)= q^m \phi(-w) + \left ( q^{m-1} \sum_{t \in {\mathbb F}_q}Z^t \right)
\sum_{u \in {\mathbb F}_q^m \setminus \{ -w\}} \phi(u).$$

The Lemma follows from the equality above and from the fact that:
$$\sum_{u \in {\mathbb F}_q^m \setminus \{ -w\}} \phi(u)=
\sum_{u \in {\mathbb F}_q^m } \phi(u) -\phi(-w)= \psi(0) -\phi(-w).$$
\end{proof}

We want to characterize the kernel of $T_{(q,m)}$, in order to do
that, we need the following lemma:

\begin{lemma}\label{pro}
A function $\phi \in {{\mathcal G}}_{(q,m)}$ verifies
$$\phi (w)\cdot \left (q-\sum_{t \in {\mathbb F}_q}Z^t\right )=0$$ for each $w \in
{\mathbb F}_q^m$ if, and only if,
$$\phi (w) = \lambda (w) \sum_{t \in {\mathbb F}_q}Z^t,$$
where $\lambda$ is a function from ${\mathbb F}_q^m$ into ${\mathbb C}$.
\end{lemma}

\begin{proof}
Let $\phi$ be given by $\phi (w)=\sum_{t \in
{\mathbb F}_q}C_t(\phi)(w)Z^t,$ then we have

$$\phi (w)\cdot \left (q-\sum_{t \in {\mathbb F}_q}Z^t\right )=
q\phi(w)-\left ( \sum_{t \in {\mathbb F}_q}C_t(\phi)(w)\right ) \left (
\sum_{t \in {\mathbb F}_q}Z^t\right ).$$

If this product is equal to zero, then
$$\phi(w)=(1/q)\left ( \sum_{t \in {\mathbb F}_q}C_t(\phi)(w)\right )
\left ( \sum_{t \in {\mathbb F}_q}Z^t\right ).$$ On the other hand a direct
computation of
$$\left (\lambda(w)\sum_{t \in {\mathbb F}_q}Z^t\right ).
\left (q-\sum_{t \in {\mathbb F}_q}Z^t\right )$$ shows the converse.
\end{proof}

\medskip

Now we can determine the kernel of $T_{(q,m)}$.

\begin{theorem}\label{kernel}
The kernel of $T_{q,m}$ is the subspace of the functions $\phi$
such that for each $w \in {\mathbb F}_q^m$
$$\phi(w)=\lambda (w) \sum_{t\in {\mathbb F}_q}Z^t$$
where $\lambda$ is any function from ${\mathbb F}_q^m$ into ${\mathbb C}$ verifying
$$\sum_{u \in {\mathbb F}_q^m} \lambda (u)=0.$$
The dimension of the kernel is $q^m-1$.
\end{theorem}

\begin{proof}
Note that if the transform of $\phi$ is the zero function, then
using the Proposition \ref{double} we get
$$\phi (w)\cdot \left (q-\sum_{t \in {\mathbb F}_q}Z^t\right )=0,$$
and by Lemma \ref{pro}
$$\phi(w)=\lambda (w) \sum_{t\in {\mathbb F}_q}Z^t.$$
Hence, for each $t\in {\mathbb F}_q$ we must have
$$C_t(\phi)(w)=\lambda(w).$$
If we denote by $\psi$ the transform of $\phi$ we know that

\begin{align*}
C_t(\psi)(w)&=\sum_{u \in {\mathbb F}_q^m}C_{\langle u,w
\rangle+t}(\phi)(u)\\
&=\sum_{u \in {\mathbb F}_q^m}\lambda (u)
\end{align*}

The result follows. Let us remark that the functions $\lambda$
such that
$$\sum_{u \in {\mathbb F}_q^m} \lambda (u)=0,$$
defines an hyperplane of the space of functions from ${\mathbb F}_q^m$ into
${\mathbb C}$ and then, the dimension of the kernel is $q^m-1$.
\end{proof}

\begin{proposition}\label{kern}
The functions 
$$\delta_a = \sum_{t \in {\mathbb F}_q} (e_{0,t}-e_{a,t})$$
with $a \in {\mathbb F}_q^m \setminus \{ 0 \}$ are a basis of the kernel
${\rm Ker}(T_{(q,m)})$, where
$$e_{u,t}(v)=\begin{cases} Z^t &\quad\text{ if $v=u$}\\
0 &\quad\text{otherwise}
\end{cases}
$$

The functions $e_{a,t}$ with 
$$(a \neq 0 \hbox{ and } t \neq 0) \hbox{ or  } (a=0) $$
is a basis of a complement of ${\rm Ker}(T_{q,m})$.
\end{proposition}
\begin{proof}
For any $a \in {\mathbb F}_q^m \setminus \{ 0 \}$ the following holds:
$$\delta_a(v)=\lambda(v) \sum_{t \in {\mathbb F}_q}Z^t,$$
with $\lambda(0)=1$, $\lambda(a)=-1$ and $\lambda(v)=0$ for the other values of $v$,
then by Theorem \ref{kernel} $\delta_a$ is in the kernel of $T_{(q,m)}$.
As the $e_{a,t}$ are linearly independent, the $\delta_a$ are linearly independent.
We conclude that the $\delta_a$ constitute a basis of ${\rm Ker}(T_{(q,m)})$.

Let $I$ be the set
$$I=\{(v,t)~|~ (v \neq 0 \hbox{ and } t \neq 0) \hbox{ or } (v=0)\}$$ and
$\phi$ the function
$$\phi= \sum_{(v,t)\in I} \lambda_{v,t} e_{v,t}.$$ 
The following holds:
$$T_{(q,m)}(\phi)(v)=\sum_{\{t|(v,t)\in I\}}\lambda_{v,t} Z^t.$$
If $v\neq 0$ then
$$T_{(q,m)}(\phi)(v)=\sum_{t\in{\mathbb F}_q^{*}}\lambda_{v,t} Z^t.$$
Then $T_{(q,m)}(\phi)(v)$ cannot be a multiple of $\sum_{t\in {\mathbb F}_q}Z^t$ unless
all the $\lambda_{v,t}$ are zero for $v\neq 0$ and in this case the
coefficient of $\sum_{t\in {\mathbb F}_q}Z^t$ is $0$. Now if $v=0$
then
$$T_{(q,m)}(\phi)(0)=\sum_{t\in{\mathbb F}_q} \lambda_{0,t} Z^t.$$
Then $T_{(q,m)}(\phi)(0)$ cannot be a multiple of $\sum_{t\in {\mathbb F}_q}Z^t$ unless
all the $\lambda_{0,t}$ have the same value $\lambda_0$ and in this case the coefficient of
$\sum_{t\in {\mathbb F}_q}Z^t$ is $\lambda_0$. Hence, if $T_{(q,m)}(\phi)(v)$ can be written
$\lambda(v)\sum_{t\in {\mathbb F}_q}Z^t$, we have
$\sum_{v\in{\mathbb F}_q^m}\lambda(v)=\lambda_0$. Then, if $\phi \in {\rm Ker}(T_{(q,m)})$,
for any $(v,t)\in I$ we have $\lambda_{v,t}=0$.
We conclude that the $q^{m+1}-(q^m -1)$ linearly independent vectors  $(e_{v,t})_{(v,t)\in I}$
constitute a basis of a complement of ${\rm Ker}(T_{(q,m)})$.
\end{proof}

\begin{corollary}\label{corokern}
The vectors $\epsilon_{v,t}=T_{(q,m)}(e_{v,t})$ with $(v,t)\in I$ are linearly independent.
They constitute a basis of the image $T_{(q,m)}({\mathcal G}_{(q,m)})$.
\end{corollary}

\section{Application to the Hamming distances from a function to all codewords
of a Generalized Reed-Muller code of order $1$}\label{sec:distances}
\subsection{System of equations satisfied by the distances of a function to all codewords}\label{secsys}
Coming back to our problem, if $f \in {{\mathcal F}}_{(q,m)}$ let us
associate to it the function $F \in {{\mathcal G}}_{(q,m)}$ defined by

$$F(u)=Z^{f(u)}.$$

The transform $T_{(q,m)}(F)$ is given by

\begin{equation}\label{nvt}
\begin{array}{ll}
T_{(q,m)}(F) (v) &= \sum _{u \in {\mathbb F}_q^m} Z^{f(u)-\langle u,v \rangle}\\
&= \sum _{t \in {\mathbb F}_q} N_{v,t}(f)Z^t,
\end{array}
\end{equation}
where $N_{v,t}(f)= \sharp\{ u\in {\mathbb F}_q^m ~|~ f(u)-\langle u,v
\rangle=t \}$ as defined before.

\begin{lemma}
 Pour any $v \in {\mathbb F}_q^m$the following formula holds:
$$\sum_{t\in {\mathbb F}_q} N_{v,t}(f)=q^m.$$
\end{lemma}

\begin{proof}
It is a direct consequence of the equalities (\ref{nvt}). Indeed
the total sum of coefficients in the first expression is $q^m$
and in the second one it is $\sum_t N_{v,t}$.
\end{proof}

As one can see, the numbers $N_{v,t}(f)$ are exactly the
coefficients of $T_{(q,m)}(F)$ where $F$ is associated to $f$ as
above.

For each $w \in {\mathbb F}_{q^m}$ we consider the linear form $L_w$
defined over ${\mathbb F}_{q^m}\times {\mathbb F}_q$ by
$$L_w(v,t)=-\langle w,v\rangle +t,$$
and for each $w \in {\mathbb F}_{q^m}$ and each $\tau \in {\mathbb F}_q$ we consider
the hyperplane $E_{w,\tau}$ of ${\mathbb F}_{q^m}\times {\mathbb F}_q$ defined by
$$E_{w,\tau}=\{ (v,t)\in {\mathbb F}_{q^m}\times {\mathbb F}_q ~|~ L_w(v,t)=\tau\}.$$

\begin{theorem}\label{sys}
Let $f \in {{\mathcal F}}_{(q,m)}$, then $N_{v,t}(f)$ are solutions of the
following linear system with $q^{m+1}$ equations on $q^{m+1}$ variables
where the equation numbered $(w,\tau)$ is:
\begin{equation*}
(w,\tau) \quad \quad \sum_{(v,t) \in E_{w,\tau}}x_{v,t}= \left \{
\begin{array}{ccc}
q^{2m-1}-q^{m-1} &  \quad if  \quad f(-w)\neq \tau \\
q^{2m-1}-q^{m-1} + q^m &  \quad if  \quad f(-w)=\tau
\end{array}
\right .
\end{equation*}
\end{theorem}

\begin{proof}
Computing $T_{(q,m)}^2(F)$, where $F=Z^{f}$ and by using the
result of Theorem \ref{double}, we obtain
$$T_{(q,m)}^2(F)(w)=q^{m-1}(q-\sum_{t\in {\mathbb F}_q}Z^t)F(-w)
+q^{m-1}\sum_{t\in{\mathbb F}_q}Z^t \sum_{u\in{\mathbb F}_{q^m}}F(u).$$
Denoting $K(Z)=\sum_{t\in {\mathbb F}_q}Z^t$ and observing that
$$K(Z) \sum_{t\in{\mathbb F}_q} \alpha_tZ^t=(\sum_{t\in{\mathbb F}_q}\alpha_t)K(Z),$$
we obtain
\begin{align*}
T_{(q,m)}^2(F)(w)&=q^mZ^{f(-w)}+K(Z)(q^{2m-1}-q^{m-1})\\
&=(q^{2m-1}-q^{m-1}+q^m)Z^{f(-w)}\\
&+ (q^{2m-1}-q^{m-1})\sum_{t \neq
f(-w)}Z^t.
\end{align*}

On the other hand, if we compute $T_{(q,m)}^2(F)$ using that
$$T_{(q,m)}(F)(v)=\sum_{t\in{\mathbb F}_{q^m}}N_{v,t}Z^t,$$
we obtain
$$T_{(q,m)}^2(F)(w)=\sum_{\tau \in{\mathbb F}_q}
\left(\sum_{(v,t)\in E_{w,\tau}}N_{v,t}\right) Z^{\tau}.$$
The theorem follows by comparing the two expressions obtained for
$T_{(q,m)}^2(F)$.
\end{proof}

\begin{remark}
The system presented in Theorem \ref{sys} has the following structure:
it is constituted by $q^m$ blocks ${\mathcal B}_w$ of $q$ equations. 
The block ${\mathcal B}_w$ contains the $q$ equations numbered $(w,\tau)$
where $w$ is fixed and $\tau$ takes the $q$ possible values in ${\mathbb F}_q$.
Each equation of a block involves $q^m$ variables, namely the variables indexed
by the points $(v,t)$ of the hyperplane $E_{w,\tau}$ of 
${\mathbb F}_{q^m}\times {\mathbb F}_q$.  The $q$ hyperplanes $E_{w,\tau}$
($w$ fixed, $\tau \in {\mathbb F}_q$) are parallel, then each variable $x_{v,t}$
is in one and only one equation of each block ${\mathcal B}_w$.
\end{remark}

Let us consider the basis defined in section \ref{sec:transform} by (\ref{base}). 
Remark that the matrix of the system (\ref{sys}) is the matrix ${\mathcal T}_{(q,m)}$
of $T_{(q,m)}$ with respect to the considered basis. Namely by construction
(see the proof of Theorem \ref{sys}), the system can be written
$${\mathcal T}_{(q,m)}X=B,$$
where $X$ is the column
$$X=\left(
\begin{array}{c}
\vdots\\
x_{v,t}\\
\vdots
\end{array}
\right),
$$
and $B$ the column 
$$B=\left(\begin{array}{c}
\vdots\\
b_{w,\tau}\\
\vdots
\end{array}
\right),
$$
where 
$$
b_{w,\tau}=\left \{
\begin{array}{ccc}
q^{2m-1}-q^{m-1} &  \quad if  \quad f(-w)\neq \tau \\
q^{2m-1}-q^{m-1} + q^m &  \quad if  \quad f(-w)=\tau
\end{array}
\right .
.$$
The system has a solution because we know that the values $N_{v,t}(f)$ constitute a solution.
But, as the linear map $T_{(q,m)}$ has a kernel, the system has not a unique solution.
However, if we add some ``normalization'' conditions we obtain the desired solution.

\begin{theorem}\label{ajout}
The numbers $N_{v,t}(f)$ are the unique solution of the system
that appears on the Theorem \ref{sys} if we join the following
$q^m$ equations
$$\sum_{t \in {\mathbb F}_q}x_{v,t}=q^m. \quad \forall v \in {\mathbb F}_{q^m}.$$
\end{theorem}

\begin{proof}
We know that any other solution is obtained from the previous
solution $\left (N_{v,t}\right)_{v,t}$ by adding an element in the
kernel of the transformation, that is any other solution has the
form $\left (\strut N_{v,t}+\lambda(v)\right)_{v,t}$ with
$\sum_v\lambda(v)=0$. For any $v$ fix, we have that $\sum_{t \in
{\mathbb F}_q}N_{v,t}=q^m$ and the result follows from it.
\end{proof}

\subsection{Transformation into a Cramer linear system}
\begin{theorem}
The system $(S)$ constructed in the following way:
\begin{enumerate}
 \item suppress from the system (\ref{sys}) the $q^m-1$ lines numbered  $(w,0)$ with $w\neq 0$,
 \item replace these equations by the $q^m-1$ equations $\sum_{t \in {\mathbb F}_q}x_{w,t}=q^m$, where $w \neq 0$,
\end{enumerate}
is a Cramer linear system and has $\left(\strut N_{v,t}(f)\right)_{v,t}$ for unique solution.
\end{theorem}

\begin{proof}
Let ${\mathcal T}_{(q,m)}$ the matrix of the original system.
The columns are the vectors $T_{(q,m)}(e_{v,t})=\epsilon_{v,t}$ decomposed
on the basis $(e_{v,t})_{(v,t)\in {\mathbb F}_q \times {\mathbb F}_q}$. 

Let us consider
the columns $(v,t)$ for which one of the two following conditions holds:
\begin{enumerate}
\item $v=0$;
\item $v\neq 0 \hbox{ and } t \neq 0$.
\end{enumerate}
Denote by $I$ these indexes.
We know by Lemma \ref{corokern} that these $q^{m+1}-(q^m-1)$ columns are linearly independent.  

Denote by $a_{(w,\tau),(v,t)}$ the coefficient of ${\mathcal T}_{(q,m)}$
which is at the line indexed by $(w,\tau)$ and the column indexed by $(v,t)$.
This coefficient is the component of $\delta_{v,t}$ on $e_{w,\tau}$, namely by Lemma \ref{evt}:
$$
a_{(w,\tau),(v,t)}=\left\{
\begin{array}{ccc}
1 & \hbox{ if } & (w,\tau)\in E_{-v,t}\\
0 & \hbox{ if } & (w,\tau)\notin E_{-v,t}
\end{array}
\right .
.
$$ 
But as the relation $(w,\tau)\in E_{-v,t}$ is equivalent to $(v,t)\in E_{w,\tau}$
we have $$a_{(w,\tau),(v,t)}=a_{(v,t),(-w,\tau)}.$$
Then the elements of line $((w,\tau)$ are the elements of the column $(-w,\tau)$
By Proposition \ref{kern} the $q^{m+1}-(q^m-1)$ lines indexed by $(w,\tau)$ where
$w \neq 0$ and $t\neq 0$,  or $w=0$, are linearly independent.

Remark that the original system has a vector space of dimension $q^m-1$ of solutions
$(x_{v,t})_{(v,t) \in {\mathbb F}_q^m\times {\mathbb F}_q}$. Adding all  equations of the system
gives the following equality:
$$\sum_{v,t} x_{v,t} =q^{2m}.$$ 
Then if we suppose that the $q^m-1$ conditions
$$\sum_{t\in{\mathbb F}_q} x_{v,t}=q^m,$$
where $v \neq 0$, are satisfied,  the last condition 
$$\sum_{t\in{\mathbb F}_q} x_{0,t}=q^m$$
is also satisfied. Now, using Theorem \ref{ajout}, we conclude that 
$(S)$ is a Cramer linear System.
\end{proof}

\begin{remark}
From the definition it follows that
$$
N_{v,t}={\rm card}\left(\{w\in {\mathbb F}_{q^m} \mid (v,t)\in E_{w,f(-w)}\}\right).
$$
So, it would be interesting to consider the arrangement of hyperplanes ${\mathcal A}(f)$,
consisting of the $q^m$ hyperplanes $E_{w,f(-w)}$ and to relate the geometric and combinatorial properties
of ${\mathcal A}(f)$ to the properties of the distance between $f$ and the affine functions.
A very simple example is the following: if the arrangement ${\mathcal A}(f)$ is centered, 
then there is a $(v,t)$
such that $N_{v,t}=q^m$ and consequently the function $f$ is affine.
\end{remark}

\end{document}